\documentclass[12pt,a4paper]{elsarticle}

\usepackage{amsmath}
\usepackage{amssymb}

\usepackage{algorithm}
\usepackage{algorithmic}
\usepackage{fullpage}

\newcommand{\G}{\mathcal{G}}
\newcommand{\N}{\mathbb{N}}

\newtheorem{theorem}{Theorem}
\newtheorem{proposition}{Proposition}
\newtheorem{lemma}{Lemma}

\newtheorem{question}{Question}
\newproof{proof}{Proof}

\begin{document}

\begin{frontmatter}
\title{Alternative polynomial-time algorithm for Bipartite Matching}
\author{Sylvain Guillemot}
\ead{guillemo@free.fr}

\begin{abstract}
If $G$ is a bipartite graph, Hall's theorem \cite{H35} gives a condition for the existence of a matching of $G$ covering one side of the bipartition. This theorem admits a well-known algorithmic proof involving the repeated search of augmenting paths. We present here an alternative algorithm, using a game-theoretic formulation of the problem. We also show how to extend this formulation to the setting of balanced hypergraphs.
\end{abstract}
\end{frontmatter}

\section{\label{s:intro}Introduction}

Let $G = (V,E)$ be a bipartite graph with bipartition $V_1, V_2$. Hall's theorem \cite{H35} states that $G$ has a matching covering $V_1$ if for each $S$ subset of $V_1$, there are at least $|S|$ vertices in $V_2$ adjacent to $S$. This theorem can be proved and extended in a number of ways (see e.g. \cite{HV50,R67,T54,CCKV96,AH00}). It has strong connections to other results about bipartite graphs such as K\"onig theorem for maximum cardinality matchings \cite{K31} and Egerv\'ary theorem for maximum weighted matchings \cite{E31}.

A well-known algorithmic proof of Hall's theorem relies on the notion of augmenting path, which can be traced back to \cite{K31,B57}. This technique actually yields a direct algorithmic proof of K\"onig theorem, as well as an indirect algorithmic proof of Egerv\'ary theorem using the Hungarian method \cite{K55}. This technique also applies to other problems more general than bipartite matching, such as in Edmonds' algorithm for nonbipartite matching \cite{E65}, Lawler's algorithm for matroid intersection \cite{L75}, and Gabow \& Stallman's algorithm for linear matroid parity \cite{GS86}. Yet, there is one situation where this technique does not seem to apply: to compute a maximum matching of a balanced hypergraph \cite{B72}. Such a matching can be computed in polynomial time by LP techniques, due to a min-max theorem of \cite{BLV70} generalizing K\"onig theorem. However, finding a purely combinatorial algorithm for this problem remains an important open question in the field, which does not seem amenable to the previous technique. This led us to seek an alternative algorithm for bipartite matching, that may extend to balanced hypergraphs.

Such an alternative algorithm is presented here. It relies on a folklore game-theoretic formulation of the problem, that is a variant of Slither \cite{A74}. We consider a two-person game played on a graph $G$ with a distinguished vertex $v_0$. Starting at $v_0$, each player in turn chooses a previously unchosen vertex, with the restriction that the sequence of vertices chosen forms a path in $G$; the first player unable to move loses. It follows from the results of \cite{A74} that a strategy for this game can be expressed in terms of maximum matchings on the underlying graph. We are concerned here with the case of a bipartite graph $G$, and in this case it turns out that the game has a memoryless strategy, called an \emph{assignment}. We give a simple polynomial-time algorithm that computes an assignment for a given pair $(G,v_0)$; we observe here that the polynomial-time termination of the algorithm is nontrivial and crucially relies on $G$ being bipartite. We are also able to extend the previous game-theoretic formulation to balanced hypergraphs; the game and the notion of assignment can be adapted to this setting, although we don't know here how to efficiently compute an assignment.

This paper is organized as follows. Section \ref{s:graph-version} deals with the case of bipartite graphs. We first define the game and the notion of assignment. We then explain how these notions provide an alternative proof of Hall's theorem. Finally, we describe a $O(n^3)$ time algorithm computing an assignment in a graph with $n$ vertices, which makes the alternative proof algorithmic. Section \ref{s:hypergraph-version} deals with the case of balanced hypergraphs. We adapt the game and the notion of assignment to this setting, and we prove non-constructively the existence of an assignment.

\section{\label{s:graph-version}A game-theoretic formulation of Hall's theorem}

\subsection{\label{ss:graph-defs} Preliminaries}

The graphs we consider are undirected, finite and simple. Let $G = (V,E)$ be a graph, where $V$ is its set of vertices and $E \subseteq [V]^2$ is its set of edges. Given $e = \{u,v\} \in E$, we denote $e = uv$. We say that $u,v$ are \emph{adjacent} if $uv \in E$, and we say that $u$ is \emph{incident} to $e$ if $u \in e$. Given $U \subseteq V$, we define the \emph{induced subgraph} $G[U] = (U,\{ uv \in E : u,v \in U \})$. We say that $G$ is \emph{bipartite} (with bipartition $V_1,V_2$) iff $V = V_1 \cup V_2$ and each edge of $G$ has the form $uv$ with $u \in V_1, v \in V_2$.

A \emph{walk} in $G$ is a sequence $W = v_1 v_2 \ldots v_k$ where $v_1,\ldots,v_k \in V$, and for each $i \in [k-1]$ we have $v_i v_{i+1} \in E$. The \emph{support} of $W$ is $V(W) = \{v_1,\ldots,v_k\}$; the \emph{length} of $W$ is $|W| = k-1$. By convention, the empty sequence $\epsilon$ is a walk of support $\emptyset$ and of length $-1$. A \emph{path} in $G$ is a walk $P = v_1 v_2 \ldots v_k$ such that: for $\{i,j\} \subseteq [k]$ we have $v_i \neq v_j$. Given $U \subseteq V$, the \emph{neighborhood of $U$} is $N_G(U) = \{ v \in V - U : \text{ there exists } uv \in E \text{ with } u \in U \}$. Given $u \in V$, we abbreviate $N_G(\{u\})$ by $N_G(u)$.

We recall Hall's theorem below. We first need the following additional definitions.
Let $S \subseteq E$ and $T \subseteq V$. We say that $S$ is a \emph{matching} of $G$ if it is formed by pairwise disjoint edges. We say that $S$ \emph{covers} $T$ iff each vertex of $T$ is incident to an edge of $S$. Clearly, if $G$ is a bipartite graph with bipartition $V_1,V_2$ having a matching covering $V_1$, it must hold that $|N_G(S)| \geq |S|$ for every $S \subseteq V_1$. Hall's theorem states that the converse is true (see \cite{HV50,R67} for alternative proofs).

\begin{theorem} \label{thm1} \cite{H35} Let $G = (V,E)$ be a bipartite graph with bipartition $V_1,V_2$. Suppose that for every $S \subseteq V_1$ we have $|N_G(S)| \geq |S|$. Then $G$ has a matching covering $V_1$.
\end{theorem}

\subsection{\label{ss:the-game} The game $\G(G',v_0)$}

Let $G = (V,E)$ be a bipartite graph with bipartition $V_1,V_2$. We introduce below a game played on a bipartite graph $G'$ derived from $G$, such that the existence of a matching covering $V_1$ translates to a winning strategy for this game. This is inspired by the Slither game of \cite{A74}.

Let $G' = (V',E')$ obtained from $G$ by adding (1) a vertex $v_1$ adjacent to all vertices of $V_1$, (2) a vertex $v_0$ adjacent to $v_1$. We consider the following two-player game $\G(G',v_0)$. A play is a sequence $P = v_0 v_1 \ldots v_k$ that is a path in $G'$. The game starts with $P = v_0$. Suppose that it is the turn of Player $s \in \{1,2\}$, and that we have $P = v_0 v_1 \ldots v_k$ with $k-1 \equiv s~(2)$. Let $Z = \{ v \in V' - V(P) : v_k v \in E' \}$. The player loses if $Z = \emptyset$; otherwise, he chooses $v \in Z$, lets $v_{k+1} = v$, and hands over to the other player. It can be shown that Player 1 has a winning strategy for this game iff $G$ has a matching covering $V_1$.

An \emph{assignment} for $(G',v_0)$ is a pair $(R,\sigma)$, where $R \subseteq V'$, $v_0 \in R$, and $\sigma : V' \rightarrow V' \cup \{\perp\}$ is such that:
\begin{itemize}
\item[$(C_1)$] for $v \in R$, if $\sigma(v) = u \in V'$ then $u \in N_{G'}(v) \cap R$ and $\sigma(u) = \perp$;
\item[$(C_2)$] for $v \in R$, if $\sigma(v) = \perp$ then for every $u \in N_{G'}(v)$ we have $u \in R$ and $\sigma(u) \neq \perp$;
\item[$(C_3)$] for $v \in R$, we have $|\sigma^{-1}(v)| \leq 1$, and $\sigma^{-1}(v_0) = \emptyset$.
\end{itemize}
Informally, $(R,\sigma)$ defines a strategy for playing the game $\G(G',v_0)$; here, $R$ is the set of reachable positions, and a position $v \in R$ is \emph{winning} if $\sigma(v) \neq \perp$, and \emph{losing} otherwise. The intuition behind this definition is as follows. Condition $(C_1)$ expresses that for a winning position $v$, \textit{some} adjacent position $u$ is losing; this information is provided by $u = \sigma(v)$. Condition $(C_2)$ expresses that for a losing position $v$, \emph{every} adjacent position $u$ is winning. Finally, Condition $(C_3)$ is necessary to ensure that the strategy does not loop on some vertex.

The following result explains how to obtain a winning strategy for either player, given an assignment.

\begin{proposition} \label{prop1} Suppose that we have $(R,\sigma)$ assignment for $(G',v_0)$.
\begin{enumerate}
\item If $\sigma(v_0) \neq \perp$, Player 1 has a winning strategy in $\G(G',v_0)$.
\item If $\sigma(v_0) = \perp$, Player 2 has a winning strategy in $\G(G',v_0)$.
\end{enumerate}
\end{proposition}

\begin{proof} We only prove (1), since (2) follows by a similar argument. We maintain the following invariant: at step $k$ of the game, the current play is $P = v_0 v_1 \ldots v_k$ with: (a) $P$ path in $G'$, (b) for an integer $i$ $(0 \leq i \leq k)$ even, we have $v_i \in R, \sigma(v_i) \neq \perp$ and ($i < k \Rightarrow v_{i+1} = \sigma(v_i)$), (c) for an integer $i$ ($0 \leq i \leq k$) odd, we have $v_i \in R, \sigma(v_i) = \perp$. Player 1 uses the following strategy: if $P = v_0 v_1 \ldots v_k$ with $k$ even, he plays $v_{k+1} = \sigma(v_k)$. We show by induction on $k$ that the invariant is preserved.

Suppose first that $k$ is odd. Let $Z = \{ v \in V' - V(P) : v_k v \in E' \}$. If $Z = \emptyset$, Player 2 loses. Suppose that $Z \neq \emptyset$ and Player 2 chooses a vertex $v_{k+1} \in Z$. We then have Point (a) at step $k+1$ by the definition of the game. Since $v_k \in R$ and $\sigma(v_k) = \perp$, by Condition $(C_2)$ we have $v_{k+1} \in R$ and $\sigma(v_{k+1}) \neq \perp$. These facts together with the induction hypothesis imply that Points (b)-(c) hold at step $k+1$.

Suppose now that $k$ is even. Player 1 then chooses $v_{k+1} = \sigma(v_k)$. Since $v_k \in R$ and $\sigma(v_k) \neq \perp$, by Condition $(C_1)$ we have $v_k v_{k+1} \in E'$, $v_{k+1} \in R$ and $\sigma(v_{k+1}) = \perp$. If we had $v_{k+1} = v_{i+1}$ for some $i < k$, we would have $i$ even and $v_k,v_i \in \sigma^{-1}(v_{k+1})$, contradicting $(C_3)$. We conclude that $P' = v_0 v_1 \ldots v_{k+1}$ is a path in $G'$, which proves Point (a) at step $k+1$. Finally, Points (b) and (c) hold at step $k+1$ by induction hypothesis and since $v_{k+1} = \sigma(v_k)$, $v_{k+1} \in R, \sigma(v_{k+1}) = \perp$.
\qed
\end{proof}

The following proposition shows how to compute a matching or an obstruction from an assignment.

\begin{proposition} \label{prop2} Suppose that we have $(R,\sigma)$ assignment for $(G',v_0)$.
\begin{enumerate}
\item If $\sigma(v_0) \neq \perp$, we can obtain a matching of $G$ covering $V_1$.
\item If $\sigma(v_0) = \perp$, we can obtain a set $S \subseteq V_1$ such that $|N_G(S)| < |S|$.
\end{enumerate}
\end{proposition}

\begin{proof} Point 1. Let $(R,\sigma)$ be an assignment for $(G',v_0)$ such that $\sigma(v_0) \neq \perp$. By $(C_1)$, we have $\sigma(v_0) = v_1$, $v_1 \in R$ and $\sigma(v_1) = \perp$.
Fix a vertex $u \in V_1$ and let $u' = \sigma(u)$. We have $u \in R$ and $u' \neq \perp$ by $(C_2)$; we have $u' \in R$ and $uu' \in E'$ by $(C_1)$; finally, by $(C_3)$ it follows that these edges are pairwise disjoint edges of $G$. We conclude that the set $M = \{ uu' : u \in V_1 \}$ is a matching of $G$ covering $V_1$. 

Point 2. Let $(R,\sigma)$ be an assignment for $(G',v_0)$ such that $\sigma(v_0) = \perp$. Let $S = \{ u \in R \cap V_1 : \sigma(u) = \perp \}$ and $T = N_G(S)$. We claim that: $|T| < |S|$.

We first show that for each $v \in T$, we have $v \in R \cap V_2$, $\sigma(v) \neq \perp$ and $\sigma(v) \in S$. Fix $v \in T$ and let $\sigma(v) = u'$ with $u' \in V' \cup  \{\perp\}$. By definition of $T$, we have $u \in S$ such that $uv \in E$. By definition of $S$, we have $u \in R \cap V_1$ and $\sigma(u) = \perp$. 
It follows by $(C_2)$ that $v \in R \cap V_2$ and $u' \neq \perp$, and it follows by $(C_1)$ that $u' \in R \cap V_1$ and $\sigma(u') = \perp$. We have thus shown that $v \in R \cap V_2$, $\sigma(v) \neq \perp$ and $\sigma(v) \in S$.

We infer by $(C_3)$ that $\sigma$ induces an injection from $T$ to $S$, and thus $|T| \leq |S|$. Let $r = \sigma(v_1)$. By $(C_2)$ and $(C_3)$, we have $r \in S$ and $\sigma^{-1}(r) = \{v_1\}$. We conclude that $|T| < |S|$.
\qed \end{proof}

In the rest of this section, we present a combinatorial algorithm for computing an assignment. Formally, we prove the following.

\begin{theorem} \label{thm2} Consider a pair $(G',v_0)$ as above, and suppose that $G'$ has $n$ vertices. We can compute in $O(n^3)$ time an assignment for $(G',v_0)$.
\end{theorem}

Combining this result with Proposition \ref{prop2}, we obtain an alternative algorithmic proof of Theorem \ref{thm1}. Using this algorithm with the Hungarian method of \cite{K55} then yields an alternative $O(n^4)$ time algorithm to compute a maximum weighted matching in a bipartite graph.

\subsection{\label{ss:algorithm}Computing an assignment}

We describe in this section the algorithm of Theorem \ref{thm2}. We will justify its correctness in Section \ref{ss:correctness}, and its running time in Section \ref{ss:running-time}.\\

We first give a high-level description of the algorithm. At a given step, we have a path $P = v_0 v_1 v_2 \ldots$ in $G'$ starting at $v_0$, a set $R \subseteq V'$, and a mapping $\sigma : V' \rightarrow V' \cup \{\perp\}$. We start with $P = v_0 v_1$, $R = \{v_0,v_1\}$ and $\sigma(u) = \perp$ for every $u \in V'$. The goal is to have at the end of the algorithm: $|P| \leq 0$ and $(R,\sigma)$ assignment for $(G',v_0)$. 

Consider a step of the algorithm, where we have $P = v_0 v_1 \ldots v_k$. Let $Z$ be the set of vertices $v \in V' - V(P)$ such that $v_k v \in E'$ and $\sigma(v) = \perp$. We have two cases.

Case 1: $Z \neq \emptyset$. We choose $v \in Z$, and we update $P = v_0 v_1 \ldots v_{k+1}$ with $v_{k+1} = v$. We add $v$ to $R$. If $\sigma^{-1}(v) = \{w\}$, we set $\sigma(w) = \perp$.

Case 2: $Z = \emptyset$. If $k > 0$, we set $\sigma(v_{k-1}) = v_k$ and $\sigma(v_k) = \perp$. If $k \leq 2$, the algorithm ends. Otherwise, we remove $v_{k-1}, v_k$ from $P$.\\

Intuitively, the algorithm explores the game tree for $\G(G',v_0)$ by identifying winning/losing positions on the fly. We maintain $P$ play of the game, $R \subseteq V'$ set of reached vertices and $\sigma : V' \rightarrow V' \cup \{\perp\}$ "partial" assignment for $(G',v_0)$. The status of a vertex $v \in V(P)$ is left undecided, while the status of a vertex $v \notin V(P)$ is either winning (if $\sigma(v) \neq \perp$) or losing (if $\sigma(v) = \perp$). At a given step, we have $P = v_0 v_1 \ldots v_k$, and we look for the next vertex $v_{k+1}$. For the position $v_k$ to be winning, we should have $v_{k+1}$ to be losing, and thus we seek $v_{k+1}$ in the set $Z$ defined above. In Case 1, we choose $v_{k+1}$ in $Z$ and we add it to $R$. In Case 2, the fact that $Z = \emptyset$ means that there is no way for $v_k$ to be winning. It follows that $v_k$ is losing, and that $v_{k-1}$ is winning by choosing $v_k$ as its next move. This is reflected by setting $\sigma(v_{k-1}) = v_k$, $\sigma(v_k) = \perp$, and by removing these vertices from $P$.\\

We show below that the above algorithm correctly computes an assignment for $(G',v_0)$. To make the proof clearer, it will be convenient to consider the implementation of the algorithm described in Algorithm \ref{algo1} below. We make the following changes:
\begin{itemize}
\item in addition to $\sigma$, we maintain a mapping $\tau$ such that whenever $\sigma(u) = v$ we have $\tau(v) = u$; this allows to test efficiently if $\sigma^{-1}(v) = \{w\}$ in Case 1.
\item in Case 1, after adding to $P$ a vertex $v$ such that $\sigma^{-1}(v) = \{w\}$ we readily add $w$ to $P$.
\end{itemize}
\begin{algorithm}[!]
\caption{$\textsc{ComputeAssignment}(G',v_0)$}
\label{algo1}
\begin{algorithmic}[1]
\STATE let $\sigma : V' \rightarrow V' \cup \{\perp\}$ such that $\sigma(u) = \perp$ for each $u \in V'$
\STATE let $\tau : V' \rightarrow V' \cup \{\perp\}$ such that $\tau(u) = \perp$ for each $u \in V'$
\STATE let $R = \{ v_0, v_1 \}$ and $P = v_0 v_1$
\WHILE{$|P| \geq 1$}
\STATE suppose that $P = v_0 v_1 \ldots v_{k}$
\STATE let $Z = \{ u \in V' - V(P) : u v_{k} \in E' \text{ and } \sigma(u) = \perp \}$
\IF{$Z = \emptyset$}
\STATE $\sigma(v_{k-1}) \leftarrow v_{k}, \sigma(v_{k}) \leftarrow \perp$
\STATE $\tau(v_{k}) \leftarrow v_{k-1}, \tau(v_{k-1}) \leftarrow \perp$
\STATE remove $v_{k},v_{k-1}$ from $P$
\ELSE
\STATE choose $z \in Z$ and let $w = \tau(z)$
\STATE append $z$ to $P$, add $z$ to $R$
\STATE if $w \neq \perp$ then append $w$ to $P$
\ENDIF
\ENDWHILE
\STATE return $(R,\sigma)$
\end{algorithmic}
\end{algorithm}

\subsection{\label{ss:correctness}Correctness}

We number the steps of the while loop by integers $0,1,\ldots$ At the end of step $s$, we denote by $P^s, R^s, \sigma^s, \tau^s$ the current values of $P, R, \sigma, \tau$. We also define $Q^s = R^s - V(P^s)$. By convention, the initial values correspond to index $s = 0$, and the first step of the loop is numbered $1$.

Consider a tuple $t = (\sigma,\tau)$ where $\sigma : V' \rightarrow V' \cup \{\perp\}$ and $\tau : V' \rightarrow V' \cup \{\perp\}$. Consider two sets $Q,R \subseteq V'$. We say that a \emph{match} of $t$ is a pair $\{u,v\} \subseteq V'$ where $\sigma(u) = v, \tau(u) = \perp, \sigma(v) = \perp$ and $\tau(v) = u$. We say that $t$ is a \emph{valid tuple} for $(Q,R)$ iff:
\begin{enumerate}
\item for $u \in V'$, if $\sigma(u) \neq \perp$ or $\tau(u) \neq \perp$ then $u \in R$;
\item we have disjoint pairs $p_1,\ldots,p_m$ such that $Q = \cup_{i = 1}^{m} p_i$ and each $p_i$ is a match of $t$.
\end{enumerate}

Proposition \ref{prop3} below states that Algorithm \ref{algo1} returns the expected result, assuming that it terminates. It relies on the following lemma which states four invariant properties of Algorithm \ref{algo1}.

\begin{lemma} \label{lem1} At the end of step $s$, we have:
\begin{itemize}
\item[(a)] $P^s = v_0 v_1 \ldots v_k$ is a path in $G'$ such that $V(P^s) \subseteq R^s$;
\item[(b)] $t^s = (\sigma^s, \tau^s)$ is a valid tuple for $(Q^s,R^s)$;
\item[(c)] for $v \in Q^s$, if $\sigma^s(v) = u \in V'$ then $u \in N_{G'}(v) \cap Q^s$ and $\sigma^s(u) = \perp$;
\item[(d)] for $v \in Q^s$, if $\sigma^s(v) = \perp$ then for every $u \in N_{G'}(v) - V(P^s)$ we have $u \in Q^s$ and $\sigma^s(u) \neq \perp$.
\end{itemize}
\end{lemma}

\begin{proof} We proceed by induction on $s$. For $s = 0$, we have $P^0 = v_0 v_1$, $R^0 = \{v_0,v_1\}$, $Q^0 = \emptyset$ and $\sigma^0(u) = \tau^0(u) = \perp$ for every $u \in V'$, hence the property holds. Suppose that the property holds at step $s-1$ and let us prove it at step $s$. Consider the set $Z^s$ obtained in Line 6.\\

Case 1: $Z^s = \emptyset$. We then have $R^s = R^{s-1}$, $Q^s = Q^{s-1} \cup \{ v_k,v_{k-1} \}$ and $v_{k-1},v_k \notin Q^{s-1}$. By induction hypothesis, $P^{s-1} = v_0 v_1 \ldots v_k$ is a path in $G'$ such that $V(P^{s-1}) \subseteq R^{s-1}$. The update of $P$ in Line 10 ensures that $P^s$ is also a path in $G'$ and that $V(P^s) \subseteq R^s$, showing (a).

Let us show (b). We first show Point (1). Consider $u \in V'$ such that $\sigma^{s}(u) \neq \perp$ or $\tau^{s}(u) \neq \perp$, we need to show that $u \in R^s$. This is clear if $u \in \{v_{k-1},v_k\}$, and otherwise it follows by induction hypothesis and by the fact that $R^s = R^{s-1}$.
We now show Point (2). By induction hypothesis, there exists disjoint pairs $p_1,\ldots,p_m$ such that $Q^{s-1} = \cup_{i = 1}^{m} p_i$ and each $p_i$ is a match of $t^{s-1}$. Let $p_{m+1} = \{v_{k-1},v_{k}\}$. By the definitions in Lines 8-9, we have $p_{m+1}$ match of $t^s$, and $p_{m+1}$ disjoint from the other pairs $p_i$. Hence, we have disjoint pairs $p_1,\ldots,p_{m+1}$ such that $Q^s = \cup_{i = 1}^{m+1} p_i$ and each $p_i$ is a match of $t^s$.

Let us show (c). Consider $v \in Q^s$ such that $\sigma^s(v) = u \in V'$. If $v \in Q^{s-1}$, we have $\sigma^{s-1}(v) = \sigma^{s}(v) = u$; by induction hypothesis we have $u \in N_{G'}(v) \cap Q^{s-1}$ and $\sigma^{s-1}(u) = \perp$; it follows that $u \in N_{G'}(v) \cap Q^s$, and $\sigma^s(u) = \sigma^{s-1}(u) = \perp$. If $v \notin Q^{s-1}$, we must have $v = v_{k-1}$ and $u = v_{k}$, and the result holds since $v_{k} \in N_{G'}(v_{k-1})$ (by a), $v_{k} \in Q^s$ and $\sigma^s(v_{k}) = \perp$.

Let us show (d). Consider $v \in Q^s$ such that $\sigma^s(v) = \perp$ and $u \in N_{G'}(v) - V(P^s)$. If $v \notin V(P^{s-1})$ and $u \notin V(P^{s-1})$, we have $\sigma^{s-1}(v) = \perp$, and thus $u \in Q^{s-1}$ and $\sigma^{s-1}(u) \neq \perp$ by induction hypothesis; we conclude that $u \in Q^s$ and $\sigma^s(u) \neq \perp$. If $v \in V(P^{s-1})$ and $u \in V(P^{s-1})$, we have $v = v_k, u = v_{k-1}$, hence $u \in Q^s$, and $\sigma^s(u) \neq \perp$. If $v \in V(P^{s-1})$ and $u \notin V(P^{s-1})$, we have $v = v_{k}$, and if $\sigma^{s-1}(u) = \perp$ we would have $u \in Z^s$, contradiction. If $v \notin V(P^{s-1})$ and $u \in V(P^{s-1})$, if we had $\sigma^{s-1}(u) = \perp$ we would have $u = v_k$ and $v \in Z^s$, contradiction.\\

Case 2: $Z^s \neq \emptyset$. We first make the following observations. By choice of $z$, we have $z \in N_{G'}(v_k) - V(P^{s-1})$ and $\sigma^{s-1}(z) = \perp$. If $w = \perp$, since $\sigma^{s-1}(z) = \perp$ and $\tau^{s-1}(z) = \perp$ we have $z \notin R^{s-1}$ by (b), and in this case we have $R^s = R^{s-1} \cup \{z\}$ and $Q^s = Q^{s-1}$. If $w \neq \perp$, we have $\sigma^{s-1}(w) = z$ and $z,w \in Q^{s-1}$ by (b), $z \in N_{G'}(w)$ by (c), and we have $R^s = R^{s-1}$ and $Q^s = Q^{s-1} - \{z,w\}$.

Let us show (a). By induction hypothesis, $P^{s-1} = v_0 \ldots v_{k}$ is a path in $G'$ with $V(P^{s-1}) \subseteq R^{s-1}$. We have seen above that $z \in N_{G'}(v_k) - V(P^{s-1})$ and $z \in R^s$. If $w = \perp$, we have $P^s = v_0 \ldots v_k z$, thus $P^s$ is a path in $G'$ such that $V(P^s) \subseteq R^s$. Suppose now that $w \neq \perp$. We have seen above that $w \in N_{G'}(z) - V(P^{s-1})$ (since $w \in Q^{s-1}$) and $w \in R^s$. We thus have $P^s = v_0 \ldots v_k z w$ path in $G'$ such that $V(P^s) \subseteq R^s$.

Let us show (b). For (1), observe that if $\sigma^s(u) \neq \perp$ or $\tau^s(u) \neq \perp$ then $\sigma^{s-1}(u) \neq \perp$ or $\tau^{s-1}(u) \neq \perp$, which implies by induction hypothesis that $u \in R^{s-1}$ and thus $u \in R^s$. Let us show (2). By induction hypothesis, we have disjoint pairs $p_1,\ldots,p_m$ such that $Q^{s-1} = \cup_{i = 1}^{m} p_i$ and each $p_i$ is a match of $t^{s-1}$. Observe that each $p_i$ is also a match of $t^s$. If $w = \perp$, the result follows since $Q^s = Q^{s-1}$. If $w \neq \perp$, we may assume that $p_m = \{z,w\}$, and the result follows since $Q^s = Q^{s-1} - \{z,w\} = \cup_{i = 1}^{m-1} p_i$.

Let us show (c). Consider $v \in Q^s$ such that $\sigma^s(v) = u \in V'$. We also have $v \in Q^{s-1}$ and $\sigma^{s-1}(v) = u \in V'$. By induction hypothesis, we have $u \in N_{G'}(v) \cap Q^{s-1}$ and $\sigma^{s-1}(u) = \perp$. We infer $\sigma^s(u) = \perp$ and we need to show that $u \in Q^s$. Suppose the contrary, we must have $w \neq \perp$ and $u \in \{z,w\}$. By (b), we have $\tau^{s-1}(u) = v$, $\tau^{s-1}(z) = w$ and $\sigma^{s-1}(w) = z$. If $u = z$, we would obtain $v = w$ and $v \in V(P^s)$, contradiction. If $u = w$, we would have $\sigma^{s-1}(u) = z$, contradiction.

Let us show (d). Consider $v \in Q^s$ such that $\sigma^s(v) = \perp$, and $u \in N_{G'}(v) - V(P^s)$. We also have $v \in Q^{s-1}$ and $u \notin V(P^{s-1})$. Since $\sigma^{s-1}(v) = \sigma^s(v) = \perp$, we have $u \in Q^{s-1}$ and $\sigma^{s-1}(u) \neq \perp$ by induction hypothesis. Since $u \notin V(P^s)$, we conclude that $u \in Q^s$ and $\sigma^s(u) = \sigma^{s-1}(u) \neq \perp$. \qed
\end{proof}

\begin{proposition} \label{prop3} Consider $(R,\sigma)$ as returned in Line 17. Then: $(R,\sigma)$ is an assignment for $(G',v_0)$.
\end{proposition}

\begin{proof} Let $s$ be the last step of the algorithm. At this step, we must have $Z^s = \emptyset$, and we have either (a) $P^s = \epsilon$ and $Q^s = R^s$, or (b) $P^s = v_0$, and $Q^s = R^s - \{v_0\}$. Clearly, we have $R \subseteq V'$ and $v_0 \in R'$. We need to show that Conditions $(C_1)$-$(C_2)$-$(C_3)$ hold for $(R,\sigma)$. Condition $(C_1)$ follows by Point (c) of Lemma \ref{lem1} applied at step $s$. Condition $(C_2)$ follows by Point (d) of Lemma \ref{lem1} applied at step $s$. Let us show condition $(C_3)$. Consider $v \in R$. If $v = v_0$, the instructions in Lines 8-10 ensure that $\sigma^{-1}(v) = \emptyset$. Suppose now that $v \neq v_0$, we then have $v \in Q^s$. By Point (b) of Lemma \ref{lem1}, we have $\sigma^{-1}(v) = \emptyset$ if $\tau^s(v) = \perp$, and $\sigma^{-1}(v) = \{u\}$ if $\tau^s(v) = u \in V'$. Thus, we have $|\sigma^{-1}(v)| \leq 1$. \qed
\end{proof}

\subsection{\label{ss:running-time}Running time}

We consider an execution of the algorithm on a graph $G'$ with $n$ vertices, and we let $S$ denote its set of steps. Here, $S$ is an initial interval of $\N$, possibly infinite. We will define a set $S_v \subseteq S$ for each $v \in V'$, and we will first give an upper bound on $|S_v|$, which will then yield an upper bound on $|S|$.

Consider a pair $p = (x,y)$ with $x,y \in V' \cup \{\perp\}$. We say that $p$ \emph{contains} $v$ if $v \in \{x,y\}$, and we let $\tilde{p} = (y,x)$. We say that a step $s$ of the algorithm is: (a) a \emph{deletion} of $p$ if step $s$ executes Lines 8-10 with $v_{k-1} = x$ and $v_k = y$; (b) an \emph{introduction} of $p$ if step $s$ executes Lines 12-14 with $z = x$ and $w = y$. 

We make the following observations. If $s$ is the deletion of $p$, then $|P^{s}| = |P^{s-1}| - 2$. If $s$ is the introduction of a pair $(x,\perp)$, then $|P^{s}| = |P^{s-1}| + 1$. If $s$ is the introduction of a pair $(x,y)$ with $y \neq \perp$, then $|P^{s}| = |P^{s-1}| + 2$. Moreover, if $s$ is the introduction of a pair $(x,\perp)$, then $x \notin R^{s-1}$ and $R^{s} = R^{s-1} \cup \{x\}$. 

Proposition \ref{prop4} gives lower and upper bounds on the number of iterations of the algorithm. We first need the following Lemma.
Given $v \in V'$, we let $S_v$ be the set of steps $s \in S$ such that $s$ is a deletion/introduction of a pair $p$ containing $v$.

\begin{lemma} \label{lem2} For every $v \in V'$, we have $|S_v| \leq 2n$.
\end{lemma}

\begin{proof} Suppose that $S_v$ contains elements $s_0 < s_1 < \ldots$ Let $I$ denote the set of indices $i \in \N$ such that $0 < 2i \leq |S_v|$. Observe that at the beginning of $s_0$ we have $v \notin V(P)$ and $\tau(v) = \perp$. Thus, $s_0$ is an introduction of $(v,\perp)$. We show by induction on $i \in I$ that there is a pair $p$ containing $v$ such that: (a) step $s_{2i-1}$ is the deletion of $p$, (b) step $s_{2i}$  is the introduction of $\tilde{p}$.

Suppose that the property is true for $i-1$, and let us prove it for $i$. Let $s = s_{2i-1}$ and $s' = s_{2i}$; since $i \in I$, these steps are defined. By induction hypothesis, $s_{2i-2}$ is an introduction of a pair containing $v$. Thus, at the beginning of step $s$ we have $v \in V(P)$, and $s$ is the deletion of a pair $p$ containing $v$. We suppose that $p = (v,u)$, since the case $p = (u,v)$ is symmetric. After step $s$ and until step $s'$, we have $v,u \notin V(P)$, $\sigma(v) = u$ and $\tau(u) = v$. Thus, step $s'$ is an introduction step, and it must introduce the pair $\tilde{p} = (u,v)$.

We define a mapping $\Phi : I \rightarrow V' - \{v\}$ as follows. Consider $i \in I$, and let $s = s_{2i-1}$ and $s' = s_{2i}$. The induction hypothesis applied for $i$ implies that $| P^{s'} | - | P^s |$ is odd. It follows that there is a step $t$ ($s < t < s'$) where $| P^t | - | P^{t-1} |$ is odd, and thus there is some vertex $v_i \in V' - \{v\}$ such that step $t$ is the introduction of $(v_i,\perp)$. We set $\Phi(i) = v_i$. By the above observation, $\Phi$ is injective, and thus $|I| \leq n-1$. We conclude that $|S_v| \leq 2|I| + 2 \leq 2n$.
\qed
\end{proof}

\begin{proposition} \label{prop4} Suppose that $G'$ has $n$ vertices. Algorithm \ref{algo1} executed on $(G',v_0)$ performs $O(n^2)$ iterations. Furthermore, this bound is tight.
\end{proposition}

\begin{proof} We first show that the algorithm performs $O(n^2)$ iterations. With the above definitions, the set of steps is $S = \{0\} \cup \bigcup_{v \in V'} S_v$, and it follows that $|S| \leq 1 + \sum_{v \in V'} |S_v|$. By Lemma \ref{lem2}, each term $|S_v|$ is upper bounded by $2n$, and thus $|S| \leq 2n^2 + 1$. We conclude that the number of iterations is $O(n^2)$.

We now show the tightness of the bound. Fix $n \in \N$ and consider the graph $G'_n$ defined as follows: (a) $G'_n$ has vertex set $V' = \{ v_0,v_1,\ldots,v_n,w_1,\ldots,w_n \}$, (b) $N_{G'}(v_0) = \{v_1\}$ and for each $i \in [n]$ we have $N_{G'_n}(v_i) = \{ w_1, \ldots, w_{n-i+1} \}$. Clearly, $G'_n$ has $2n+1$ vertices. It can be shown that Algorithm \ref{algo1} executed on $(G'_n,v_0)$ performs $n^2 + 1$ iterations, hence the result. \qed
\end{proof}

From Proposition \ref{prop4}, we infer that Algorithm \ref{algo1} can be implemented in $O(n^3)$ time. Indeed, assume that $G'$ is represented by an adjacency matrix indexed by $n$ vertices, and that $\sigma,\tau$, $R$ and $P$ are represented by arrays of length $n$; with this representation, each iteration of the while loop takes $O(n)$ time, and thus the total running time is at most $O(n^3)$. Together with Proposition \ref{prop3}, this completes the proof of Theorem \ref{thm2}.

\section{\label{s:hypergraph-version}An extension to balanced hypergraphs}

\subsection{\label{ss:hypergraph-defs}Preliminaries}

It will be convenient for us to represent a hypergraph by its incidence graph. Thus, we define a \emph{hypergraph} as a bipartite graph $H = (W,F)$ with bipartition $W = V \cup E$. To avoid confusion, an element of $V$ will be called a \emph{hypervertex}, and an element of $E$ will be called a \emph{hyperedge}. Given $u \in V$ and $e \in E$, we say that $u$ is \emph{incident} to $e$ if $ue \in F$. 
Aside from this, the definitions introduced for graphs in Section \ref{ss:graph-defs} carry over to hypergraphs. 

We recall that a \emph{path} in $H$ is a walk $P = z_1 z_2 \ldots z_k$ such that: for $\{i,j\} \subseteq [k]$ we have $z_i \neq z_j$; note that $P$ must alternate between $V$ and $E$. 
A \emph{cycle} in $H$ is a walk $C = z_1 z_2 \ldots z_k$ such that: (a) $z_1 = z_k$ and (b) for $\{i,j\} \subseteq [k]$ we have $z_i \neq z_j$ unless $\{i,j\} = \{1,k\}$. 
If $P = z_1 z_2 \ldots z_k$ is a path or a cycle in $H$, we say that $P$ is \emph{strong} iff $H[V(P)]$ contains exactly the edges $\{z_i,z_{i+1}\}$ ($1 \leq i < k$).

Fix $S,T \subseteq  W$. We say that $S$ \emph{covers} $T$ (in $H$) iff for every $u \in T$, we have $|N_H(u) \cap S| \geq 1$. We say that $S$ \emph{splits} $T$ (in $H$) iff for every $u \in T$, we have $|N_H(u) \cap S| \leq 1$. A \emph{matching} of $H$ is a set $M \subseteq E$ such that $M$ splits $V$ in $H$. An \emph{independent} of $H$ is a set $S \subseteq V$ such that $S$ splits $E$ in $H$. A \emph{transversal} of $H$ is a set $T \subseteq V$ such that $T$ covers $E$ in $H$. We let $\nu(H)$ denote the maximum cardinality of a matching of $H$, and we let $\tau(H)$ denote the minimum cardinality of a transversal of $H$. 

We say that $H$ is \emph{balanced} iff it has no strong cycle of length $4k+2$ for an integer $k \geq 1$ \cite{B72}. 
We have the following characterization of balanced hypergraphs due to \cite{BLV70} (see also \cite{L72,ST16} for alternative combinatorial proofs).

\begin{theorem} \label{thm3} \cite{BLV70} A hypergraph $H$ is balanced iff for every $H'$ partial subhypergraph of $H$, it holds that $\nu(H') = \tau(H')$.
\end{theorem}

\subsection{The game $\G(H',v_0)$}

Let $H$ be a balanced hypergraph with bipartition $W = V \cup E$, and let $U$ be an independent transversal of $H$. We adapt to the setting of hypergraphs the game $\G(G',v_0)$ and the corresponding notion of assignment seen in Section \ref{s:graph-version}.

We augment $H$ to a hypergraph $H'$ by adding (1) two hypervertices $v_0,v_1$, (2) a hyperedge $e_0$ with $N_{H'}(e_0) = \{v_0,v_1\}$, (3) for each $v \in U$ a hyperedge $f_v$ with $N_{H'}(f_v) = \{v_1,v\}$. 
We consider the following two-player game $\G(H',v_0)$.  A play of the game is a sequence $P = v_0 e_0 v_1 e_1 \ldots v_k e_k$ that is a strong path in $H'$. The game starts with $P = v_0 e_0$. Suppose that it is the turn of Player $s \in \{1,2\}$ and that we have $P = v_0 e_0 v_1 e_1 \ldots v_k e_k$ with $k-1 \equiv s~(2)$. If possible, the player (1) chooses $v \in V', e \in E'$ such that $P v e$ is a strong path in $H'$, (2) appends $v_{k+1} = v$ and $e_{k+1} = e$ to $P$ and (3) hands over to the other player.

Consider a function $\sigma : V' \rightarrow E' \cup \{\perp\}$. Given $u \in V'$, we let $N_{\sigma}(u) = \{ v \in V' - u : \sigma(v) \text{ is a hyperedge incident to } u \text{ in } H' \}$. We say that $(R,\sigma)$ is an \emph{assignment} for $(H',v_0)$ iff $R \subseteq V'$, $v_0 \in R$, and $\sigma : V' \rightarrow E' \cup \{\perp\}$ are such that:
\begin{itemize}
\item[($C_1$)] for $v \in R$, if $\sigma(v) = e \in E'$ then $e \in N_{H'}(v)$ and for every $u \in N_{H'}(e) - \{v\}$ we have $u \in R$ and $\sigma(u) = \perp$;
\item[($C_2$)] for $v \in R$, if $\sigma(v) = \perp$ then for every $e \in N_{H'}(v)$, there exists $u \in N_{H'}(e)$ such that $u \in R$ and $\sigma(u) \neq \perp$;
\item[($C_3$)] for $v \in R$ we have $|N_{\sigma}(v)| \leq 1$, and $N_{\sigma}(v_0) = \emptyset$.
\end{itemize}
Fix $(R,\sigma)$ assignment for $(H',v_0)$ such that $\sigma(v_0) = \perp$, and let $P = v_0 e_0 \ldots v_k e_k$ be a path in $H'$. We say that $P$ is \emph{compatible with $(R,\sigma)$} iff (a) for an integer $i$ ($0 \leq i \leq k$) odd, we have $v_i \in R$, $\sigma(v_i) \neq \perp$ and $e_i = \sigma(v_i)$, (b) for an integer $i$ ($0 \leq i \leq k$) even, we have $v_i \in R$, $\sigma(v_i) = \perp$, (c) $P$ is a strong path in $H'$.

Proposition \ref{prop5} below provides a winning strategy for either player, given an assignment. Its proof relies on the following lemma.

\begin{lemma} \label{lem3} Suppose that we have $(R,\sigma)$ assignment for $(H',v_0)$, and suppose that we have a path $P = v_0 e_0 \ldots v_k e_k$ compatible with $(R,\sigma)$, and $k$ even. Consider $v_{k+1} \in N_H(e_k) - \{v_k\}$ such that $\sigma(v_{k+1}) \neq \perp$, and let $e_{k+1} = \sigma(v_{k+1})$. Then $P' = P v_{k+1} e_{k+1}$ is a strong path in $H'$.
\end{lemma}

\begin{proof} Observe that $v_{k+1}$ is distinct from the other vertices $v_i$: we have $v_{k+1} \neq v_k$ by choice of $v_{k+1}$; for $i < k$, we have $v_{k+1} \neq v_i$ since the hyperedge $e_k$ is not incident to $v_i$ ($P$ strong path). Observe also that $e_{k+1}$ is distinct from the other hyperedges $e_i$: if we had $e_{k+1} = e_i$, we would have $j \in \{i-1,i\}$ such that $j$ odd and $e_{k+1}$ incident to $v_{j+1}$; since $\sigma(v_{j+1}) = \perp$ and $v_{j+1}$ is incident to $e_j = \sigma(v_j)$ and $e_{k+1} = \sigma(v_{k+1})$, we would obtain $v_j, v_{k+1} \in N_{\sigma}(v_{j+1})$, impossible by $(C_3)$.

Since $P$ is a strong path in $H'$, it remains to show that: (a) $v_{k+1}$ is not incident to a hyperedge $e_i$ ($i < k$); (b) $e_{k+1}$ is not incident to a hypervertex $v_i$ ($i \leq k$). 

Point (a): suppose by way of contradiction that $v_{k+1}$ is incident to a hyperedge $e_i$ ($i < k$), and let us choose such an $i$ maximum. If $i$ is odd, we have $\sigma(v_i) = e_i$, $v_{k+1} \in N_{H'}(e_i) - \{v_i\}$ and $\sigma(v_{k+1}) \neq \perp$, which contradicts $(C_1)$. If $i$ is even, we obtain that $C = v_{k+1} e_i v_{i+1} \ldots e_k v_{k+1}$ is a strong cycle of length $2(k-i) + 2$ in $H$, impossible. 

Point (b): suppose by way of contradiction that $e_{k+1}$ is incident to a hypervertex $v_i$ ($i \leq k$), and let us choose such an $i$ maximum. If $i$ is odd, we have $\sigma(v_{k+1}) = e_{k+1}$, $v_i \in N_{H'}(e_{k+1}) - \{v_{k+1}\}$ and $\sigma(v_i) \neq \perp$, which contradicts $(C_1)$. If $i$ is even, we have $\sigma(v_i) = \perp$ and $v_{i-1}, v_{k+1} \in N_{\sigma}(v_i)$, which contradicts $(C_3)$.
\qed
\end{proof}

\begin{proposition} \label{prop5} Suppose that we have $(R,\sigma)$ assignment for $(H',v_0)$.
\begin{enumerate}
\item If $\sigma(v_0) = \perp$, Player 1 has a winning strategy in $\G(H',v_0)$.
\item If $\sigma(v_0) \neq \perp$, Player 2 has a winning strategy in $\G(H',v_0)$.
\end{enumerate}
\end{proposition}

\begin{proof} We only prove (1), since (2) follows by a similar argument. We maintain the following invariant: at step $k$ of the game, the current play is $P = v_0 e_0 v_1 e_1 \ldots v_k e_k$ with $P$ path in $H'$ compatible with $(R,\sigma)$. Player 1 uses the following strategy: if $P = v_0 e_0 v_1 e_1 \ldots v_k e_k$ with $k$ even, he chooses $v \in N_{H'}(e_k) - \{v_k\}$ such that $\sigma(v) \neq \perp$, and he plays $v_{k+1} = v$ and $e_{k+1} = \sigma(v)$. We show by induction on $k$ that the invariant is preserved.

Suppose first that $k$ is odd. We have $P = v_0 e_0 \ldots v_k e_k$ path in $H'$. Player 2 then chooses $v_{k+1}, e_{k+1}$ such that $P' = P v_{k+1} e_{k+1}$ is a strong path in $H'$. We need to show that $P'$ is compatible with $(R,\sigma)$. Points (a)-(b) hold by induction hypothesis and since $v_k \in R, \sigma(v_k) = e_k$ imply $v_{k+1} \in R$ and $\sigma(v_{k+1}) = \perp$. Point (c) holds by choice of $v_{k+1},e_{k+1}$. 

Suppose now that $k$ is even. We have $P = v_0 e_0 \ldots v_k e_k$ path in $H'$. Player 1 then chooses $v \in N_{H'}(e_k) - \{v_k\}$ such that $\sigma(v) \neq \perp$, and plays $v_{k+1} = v$ and $e_{k+1} = \sigma(v)$. Observe that this is possible: by induction hypothesis, we have $v_k \in R$ and $\sigma(v_k) = \perp$, and by $(C_2)$ we find such a $v$ and we have $v \in R$. We need to show that $P' = P v_{k+1} e_{k+1}$ is compatible with $(R,\sigma)$. Points (a)-(b) hold by induction hypothesis and since $v_{k+1} \in R, e_{k+1} = \sigma(v_{k+1})$. Point (c) follows from Lemma \ref{lem3}.
\qed
\end{proof}

Building on Proposition \ref{prop5}, we now prove the following theorem. 

\begin{theorem} \label{thm4} Let $H$ be a balanced hypergraph with independent transversal $U$, and let $(H',v_0)$ constructed as above from $(H,U)$. Then: $H$ has a matching covering $U$ iff Player 2 has a winning strategy for $\G(H',v_0)$.
\end{theorem}

The proof of the Theorem proceeds by constructing an assignment $(R,\sigma)$ for $(H',v_0)$ in either case, whether or not $H$ has a matching covering $U$. This is the purpose of Proposition \ref{prop6} below.

\begin{proposition} \label{prop6} The following holds.
\begin{enumerate}
\item If $H$ has a matching covering $U$, we can find $(R,\sigma)$ assignment for $(H',v_0)$ such that $\sigma(v_0) \neq \perp$.
\item If $H$ has no matching covering $U$, we can find $(R,\sigma)$ assignment for $(H',v_0)$ such that $\sigma(v_0) = \perp$.
\end{enumerate}
\end{proposition}

\begin{proof}[Sketch] Point 1. Let $M$ be a matching covering $U$. We define $\sigma : V' \rightarrow E' \cup \{\perp\}$ as follows. Consider $v \in V'$. If $v = v_0$, we let $\sigma(v) = e_0$. If $v \in U$, we let $\sigma(v)$ be the unique hyperedge of $M$ incident to $v$. Otherwise, we let $\sigma(v) = \perp$. By definition, we have $\sigma(v_0) \neq \perp$. It is easily checked that $(V',\sigma)$ is an assignment for $(H',v_0)$.

Point 2. Suppose that $H$ has no matching covering $U$. Let $M$ be a maximum matching of $H$, and let $T$ be a transversal of $H$ of minimum cardinality. By Theorem \ref{thm3}, we have $|M| = |T|$. By the assumption, we find $r \in U$ not covered by $M$. We define $\sigma : V' \rightarrow E' \cup \{\perp\}$ as follows. We let $\sigma(v_0) = \perp$ and $\sigma(v_1) = f_r$. Consider $v \in V$. If $v \notin T$, we let $\sigma(v) = \perp$. If $v \in T$, we let $\sigma(v)$ be the unique hyperedge of $M$ incident to $v$. By definition, we have $\sigma(v_0) = \perp$. It is easily checked that $(V',\sigma)$ is an assignment for $(H',v_0)$, with points $(C_1)-(C_3)$ following from the complementary slackness conditions for $(T,M)$.
\qed
\end{proof}

The proof of Theorem \ref{thm4} then follows from Propositions \ref{prop5} and \ref{prop6}. Note that the latter result implies that there always exists an assignment for $(H',v_0)$. However, the proof does not give a direct way to construct it since it relies on Theorem \ref{thm3}. This leads us to the following open question.

\begin{question} \label{q1} Is there a combinatorial algorithm to compute an assignment for $(H',v_0)$?
\end{question}

\bibliography{paper} 
\bibliographystyle{ieeetr}

\end{document}